\documentclass[onecolumn]{IEEEtran}
\usepackage{amsfonts}
\usepackage{mathrsfs}
\usepackage{amsmath}
\usepackage{amssymb,bm}
\usepackage{graphicx}
\usepackage{epic}
\usepackage{epstopdf}
\usepackage{color}
\usepackage{cite}
\usepackage{comment}
\usepackage{mathtools}
\usepackage{enumerate}
\usepackage{lineno}
\usepackage{url}
\usepackage{float}
\usepackage{color}
\usepackage{bbm}
 \usepackage{booktabs}
\usepackage{tikz}
\usepackage{hyperref}
\usepackage[justification=centering]{caption}
\usepackage{algorithm}
\usepackage{algorithmic}
\usepackage{amsthm}
\usepackage{setspace}

\newtheorem{claim}{Claim}[section]
\newtheorem{theorem}{Theorem}[section]

\newtheorem{definition}[theorem]{Definition}

\newtheorem{construction}[theorem]{Construction}

\newtheorem{lemma}[theorem]{Lemma}

\newtheorem{observation}[theorem]{Observation}

\begin{document}
	
\title{On the Fixed-Length-Burst Levenshtein Ball with Unit Radius}

\author{Yuanxiao Xi, Yubo Sun, and Gennian Ge
    \thanks{This research was supported by the National Key Research and Development Program of China under Grant 2020YFA0712100, the National Natural Science Foundation of China under Grant 12231014, and Beijing Scholars Program.}

    \thanks{Y. Xi ({\tt yuanxiao\_xi@126.com}) is with the School of Mathematical Sciences, Zhejiang University, Hangzhou 310027, Zhejiang, China.}

    \thanks{Y. Sun ({\tt 2200502135@cnu.edu.cn}) and G. Ge ({\tt gnge@zju.edu.cn}) are with the School of Mathematical Sciences, Capital Normal University, Beijing 100048, China.}
}

\date{}
\maketitle
	
\begin{abstract}
Consider a length-$n$ sequence $\bm{x}$ over a $q$-ary alphabet. The \emph{fixed-length Levenshtein ball} $\mathcal{L}_t(\bm{x})$ of radius $t$ encompasses all length-$n$ $q$-ary sequences that can be derived from $\bm{x}$ by performing $t$ deletions followed by $t$ insertions. Analyzing the size and structure of these balls presents significant challenges in combinatorial coding theory.
Recent studies have successfully characterized fixed-length Levenshtein balls in the context of a single deletion and a single insertion. These works have derived explicit formulas for various key metrics, including the exact size of the balls, extremal bounds (minimum and maximum sizes), as well as expected sizes and their concentration properties. However, the general case involving an arbitrary number of $t$ deletions and $t$ insertions $(t>1)$ remains largely uninvestigated.
This work systematically examines fixed-length Levenshtein balls with multiple deletions and insertions, focusing specifically on \emph{fixed-length burst Levenshtein balls}, where deletions occur consecutively, as do insertions. We provide comprehensive solutions for explicit cardinality formulas, extremal bounds (minimum and maximum sizes), expected size, and concentration properties surrounding the expected value.
\end{abstract}

	
\section{Introduction}

Synchronization errors, particularly deletions and insertions, represent a significant class of channel impairments encountered in various applications, including DNA storage \cite{Organick-18-background}, racetrack memory \cite{Chee-18-IT-RM}, and document synchronization \cite{Cheng-19-ICALP-burst}. The systematic study of synchronization error-correcting codes originated with Levenshtein's foundational work \cite{Levenshtein-66-SPD-1D} in the 1960s.
To characterize the limits of optimal synchronization codes, determining the cardinalities of deletion and insertion balls is crucial.
For insertion balls, Levenshtein's seminal work \cite{Levenshtein-66-SPD-1D} completely determined their size, establishing that this size is independent of the center. In contrast, the size of deletion balls depends on the center, making the exact determination of their size considerably more challenging.
Levenshtein \cite{Levenshtein-66-SPD-1D} provided both lower and upper bounds on the size of deletion balls by analyzing the number of runs in the center.
Calabi and Hartnett \cite{Calabi-69} and Hirschberg and Regnier \cite{Hirschberg-99} offered recursive expressions for the maximum size of deletion balls, demonstrating that the deletion ball centered at a cyclic sequence (whose $i$-th entry is $i$ modulo the alphabet size) can achieve this maximum size. Additionally, Hirschberg and Regnier \cite{Hirschberg-99} improved the lower bound on the size of deletion balls. Subsequently, Mercier et al. \cite{Mercier-08} derived explicit formulas for the size of deletion balls with small radii. Most recently, Liron and Langberg \cite{Liron-15} further tightened both the upper and lower bounds on the sizes of deletion balls through a structural analysis of sequences containing a fixed number of runs.

The concurrent occurrence of deletions and insertions necessitates analysis beyond isolated deletion/insertion ball studies. Sala and Dolecek~\cite{Sala-13-ISIT-ball} pioneered the investigation of \emph{fixed-length Levenshtein ball} $\mathcal{L}_t(\bm{x})$, defined as the set of sequences obtainable from the center $\bm{x}$ through exactly $t$ deletions followed by $t$ insertions.
Characterizing these balls presents significant challenges, as even deletion ball analysis alone remains nontrivial.
The foundational work by Sala and Dolecek~\cite{Sala-13-ISIT-ball} established a general upper bound on ball sizes and derived an exact closed-form expression for radius one configurations.
Bar-Lev et al.~\cite{Bar-Lev-22-IT-ball} subsequently conducted a comprehensive analysis of radius one balls by quantifying their extremal characteristics, specifically the maximum, minimum, and expected cardinalities.
Building upon these deterministic results, Wang and Wang~\cite{Wang-24-DCC-ball} further revealed that these radius one ball sizes concentrate sharply around their expectation with high probability.
Most recently, He and Ye~\cite{He-23-ISIT-ball} extended this probabilistic analysis to radius-two scenarios, establishing analogous concentration properties.

Burst errors manifest across diverse domains, including DNA storage \cite{Organick-18-background}, racetrack memory \cite{Chee-18-IT-RM}, and document synchronization \cite{Cheng-19-ICALP-burst}. Theoretical analyses of burst error balls have been developed in \cite{Lan-25, Schoeny-17-IT-BD, Levenshtein-70-BD, Sun-25-IT-BD, Sun-23-IT-BDR, Wang-23-IT-BD}. Unlike isolated errors, burst errors exhibit spatial continuity, where all errors occur consecutively. While one may intuitively assume minimal distinction between single error correction and single burst error correction, the latter presents substantially greater complexity. This complexity is highlighted by the more than fifty-year gap between Levenshtein's optimal single-deletion correcting codes \cite{Levenshtein-66-SPD-1D} and the first optimal burst-deletion correcting codes \cite{Sun-25-IT-BD}.

In this paper, we investigate fixed-length Levenshtein balls under burst operations, termed \emph{fixed-length $b$-burst Levenshtein balls} $\mathcal{L}_t^b(\bm{x})$, defined as the set of sequences obtainable from the center $\bm{x}$ through $t$ bursts of deletions followed by $t$ bursts of insertions, each burst having a fixed length $b$. We provide a comprehensive study of fixed-length $b$-burst Levenshtein balls with radius one. Our contributions are threefold:
\begin{itemize}
  \item Establishing an exact closed-form expression for the size of fixed-length $b$-burst Levenshtein balls;
  \item Quantifying the extremal characteristics of fixed-length $b$-burst Levenshtein balls, specifically the maximum, minimum, and expected cardinalities;
  \item Revealing that these ball sizes concentrate sharply around their expectation with high probability.
\end{itemize}

The remainder of this paper is organized as follows. Section~\ref{sec:pre} introduces the relevant notations, definitions, and important tools used throughout the paper. Section~\ref{sec:main} summarizes our main contributions and provides a brief proof of the minimum size of fixed-length $b$-burst Levenshtein balls with radius one. Sections~\ref{sec:explicit}, \ref{sec:maximum}, \ref{sec:average}, and \ref{sec:concentration} present the retailed proofs regarding the explicit size, maximum size, average size, and concentration property of fixed-length $b$-burst Levenshtein balls with radius one, respectively. 

\section{Preliminaries}\label{sec:pre}

\subsection{Notations}

For two integers $i$ and $j$, let $[i, j]$ denote the integer interval $\{i, i+1, \ldots, j\}$ when $i \leq j$, and the empty set otherwise. Let $\Sigma_q$ denote the alphabet set $[0, q-1]$, and let $\Sigma_q^n$ denote the set of all sequences of length $n$ over $\Sigma_q$. For $\boldsymbol{x} \in \Sigma_q^n$, we use $x_i$ to refer to its $i$-th entry for $i \in [1, n]$, and express $\boldsymbol{x}$ as $x_1 x_2 \cdots x_n$. For another $\bm{y} \in \Sigma_q^m$, the \emph{concatenation} of $\bm{x}$ and $\bm{y}$ is denoted by $\bm{x}\bm{y}\triangleq x_1 \cdots x_n y_1 \cdots y_m$.
If there exists a set $\mathcal{I} = \{i_1, i_2, \ldots, i_m\}$ with $1 \leq i_1 < i_2 < \cdots < i_m \leq n$ such that $\bm{y} = \bm{x}_{\mathcal{I}} \triangleq x_{i_1} x_{i_2} \cdots x_{i_m}$, we say that $\bm{y}$ is a \emph{subsequence} of $\bm{x}$. In particular, when $\mathcal{I}$ is an interval, $\bm{y}$ is called a \emph{substring} of $\bm{x}$.

For two positive integers $m < n$, a sequence $\bm{x} \in \Sigma_q^n$ is called \emph{$m$-periodic} if $x_i = x_{i+m}$ for $i \in [1,n-m]$. Note that any sequence of length at most $m$ is regarded as $m$-periodic. A \emph{$b$-run} in a sequence is a substring that is $b$-periodic with maximal length, and we use $r_b(\bm{x})$ to denote the number of $b$-runs in $\bm{x}$.
In particular, when \(b = 1\), we abbreviate a \(b\)-run and \(r_b(\bm{x})\) as a \emph{run} and \(r(\bm{x})\), respectively.
An \emph{alternating segment} in a sequence is a substring that is $2$-periodic but not $1$-periodic with maximal length, and we use $\alpha(\bm{x})$ to denote the number of alternating segments in $\bm{x}$.

Let $X$ be a set or an interval, we use $|X|$ to refer to its size. Let $A$ be an event, we define its indicator function as $\mathbbm{1}_A$, where $\mathbbm{1}_A = 1$ if $A$ occurs, and $\mathbbm{1}_A = 0$ otherwise.

\subsection{Error Models}

Let \(n \geq b + 1\) with \(b \geq 1\). We say that \(\boldsymbol{x} \in \Sigma_q^n\) suffers \emph{a burst of \(b\) deletions} at position \(i\), where \(1 \leq i \leq n - b + 1\), if the resultant sequence is \(x_1 \cdots x_{i-1} x_{i+b} \cdots x_n \in \Sigma_q^{n-b}\). When \(n \geq bt + 1\), where \(t\) is a positive integer, the \emph{b-burst-deletion ball} of radius \(t\) centered at \(\boldsymbol{x}\), denoted as \(\mathcal{D}_{t}^{b}(\boldsymbol{x})\), represents the set of all sequences in \(\Sigma_q^{n-bt}\) obtainable from \(\boldsymbol{x}\) after \(t\) bursts of deletions, each of length $b$.
Similarly, we say that \(\boldsymbol{x}\) suffers \emph{a burst of \(b\) insertions} at position \(i\), where \(1 \leq i \leq n + 1\), if the resultant sequence is of the form \(x_1 \cdots x_{i-1} y_1 \cdots y_b  x_{i} \cdots x_n \in \Sigma_q^{n+b}\), where \(y_1 \cdots y_b\) denotes the inserted segment of length \(b\). The \emph{b-burst-insertion ball} of radius \(t\) centered at \(\boldsymbol{x}\) is defined as \(\mathcal{I}_{t}^{b}(\boldsymbol{x})\), representing the set of all sequences in \(\Sigma_q^{n+bt}\) that can be obtained from \(\boldsymbol{x}\) after \(t\) bursts of insertions, each of length exactly \(b\).

Let \(n \geq bt + 1\) with \(b, t \geq 1\). The \emph{fixed-length $b$-burst Levenshtein ball} of radius \(t\) centered at the sequence \(\bm{x} \in \Sigma_q^n\), denoted as \(\mathcal{L}_t^b(\bm{x})\), represents the set of sequences in \(\Sigma_q^n\) that can be obtained from \(\bm{x}\) through \(t\) bursts of deletions followed by \(t\) bursts of insertions, where each burst has a fixed length \(b\).

\begin{observation}\label{order}
The order of burst-deletion and burst-insertion operations is commutative concerning the final result. Therefore, without loss of generality, we adopt the convention that burst-deletions precede burst-insertions.
\end{observation}

\begin{IEEEproof}
Assume that \(\bm{x}\) undergoes a burst of \(b\) deletions at position \(i\) followed by a burst of \(b\) insertions at position \(j\), resulting in \(\bm{y}\). If \(i \leq j\), then \(\bm{y}\) can also be obtained from \(\bm{x}\) by first performing a burst of \(b\) insertions at position \(j + b\) followed by a burst of \(b\) deletions at position \(i\). If \(i > j\), then \(\bm{y}\) can be achieved from \(\bm{x}\) by first performing a burst of \(b\) insertions at position \(j\) followed by a burst of \(b\) deletions at position \(i + b\). The argument for the other direction can be established similarly. This completes the proof.
\end{IEEEproof}

\subsection{Useful Conclusions}

We first review the size of a burst-deletion/insertion ball.

\begin{lemma}\cite[Claim 3.1]{Sun-23-IT-BDR}\label{lem:del_ball}
    Let \( \mathcal{K} = [k_1,k_2] \subseteq [1,n] \) with \( k_2 - k_1 \geq b \geq 1 \) and \( n \geq b + 1 \). For any sequence \( \bm{x} \in \Sigma_q^n \), \( \bm{x}_{\mathcal{K}} \) is a $b$-run of \( \bm{x} \) if and only if \(\bm{x}_{[1,n]\setminus [k,k+b-1]} = \bm{x}_{[1,n]\setminus [k',k'+b-1]}\)
    for any \( k, k' \in [k_1,k_2-b+1] \).
\end{lemma}

This lemma establishes the expression for the size of a burst-deletion ball, for which we formalize as follows.
\begin{lemma}\cite[Lemma IV.4]{Lan-25}\label{del_ball}
For any \( \bm{x} \in \Sigma_q^n \) and $b\geq 1$, it holds that $\big|D_1^b(\bm{x})\big| = r_b(\bm{x})= 1+ |\{i\in [1,n-b]: x_i\neq x_{i+b}\}| \leq n-b+1$.
\end{lemma}

For the size of a burst-insertion ball, Lan et al.~\cite{Lan-25} derived the following conclusion.
\begin{lemma}\cite[Theorem III.2]{Lan-25}\label{lem:ins_ball}
For any \( \bm{x} \in \Sigma_q^n \) and $b,t\geq 1$, it holds that $\big|I_t^b(\bm{x})\big| = q^{t(b-1)}\sum_{i=0}^{t}\binom{n+t}{i}(q-1)^i$.
\end{lemma}

We now review the intersection between burst-insertion balls.

\begin{lemma}\cite[Claims 3.2 and 3.3]{Sun-23-IT-BDR}\label{num_insertion}
For two distinct sequences \( \bm{x} = \bm{u}\bm{v}\bm{w} \) and \( \bm{y} = \bm{u}\bm{v'}\bm{w} \) in \( \Sigma_q^n \), where \( \bm{v} \) and \( \bm{v}' \) differ in both their initial and terminal symbols, the following holds:
\begin{itemize}
    \item If \( |\bm{v}|=d \geq b + 1 \), then
    \[
    I_1^b(\bm{x}) \cap I_1^b(\bm{y}) \subseteq \{ \bm{u}\bm{v'}_{[1,b]}\bm{v}\bm{w}, \bm{u}\bm{v}_{[1,b]}\bm{v'}\bm{w} \}.
    \]
    More precisely, it holds that
    \begin{itemize}
        \item $\boldsymbol{u} \boldsymbol{v}_{[1,b]}' \boldsymbol{v} \boldsymbol{w} \in I_b(\boldsymbol{x}) \cap I_b(\boldsymbol{y})$ if and only if $\boldsymbol{v}_{[1,d-b]}= \boldsymbol{v}_{[b+1,d]}'$;
        \item $\boldsymbol{u} \boldsymbol{v}_{[1,b]}\boldsymbol{v}' \boldsymbol{w} \in I_b(\boldsymbol{x}) \cap I_b(\boldsymbol{y})$ if and only if $\boldsymbol{v}_{[b+1,d]}= \boldsymbol{v}_{[1,d-b]}'$.
    \end{itemize}

    \item If \( |\bm{v}| = d \leq b \), then
    \[
    I_1^b(\bm{x}) \cap I_1^b(\bm{y}) = \{ \bm{u}\bm{v}\bar{\bm{v}}\bm{v'}\bm{w}, \bm{u}\bm{v'}\bar{\bm{v}}\bm{v}\bm{w} : \bar{\bm{v}} \in \Sigma_q^{b-d} \}.
    \]
\end{itemize}
\end{lemma}

\section{Main Contributions}\label{sec:main}

In this section, we present our main contributions regarding to fixed-length $b$-burst Levenshtein balls.
Recall that $\mathcal{L}_t^b(\bm{x})$ denotes the fixed-length $b$-burst Levenshtein ball with radius $t$ centered at sequence $\bm{x}\in \Sigma_q^n$.

\subsection{Minimum Size}

In \cite{Bar-Lev-22-IT-ball}, Bar-Lev et al. demonstrated that for any center \(\bm{x} \in \Sigma_q^n\) and radius \(t\), the Hamming ball is a subset of the fixed-length Levenshtein ball. Moreover, the Hamming ball equals the fixed-length Levenshtein ball if the number of runs in the center is exactly one. One may generalize the definition of a Hamming ball to the concept of a $b$-burst Hamming ball and establish a similar conclusion regarding $b$-burst Hamming balls and fixed-length $b$-burst Levenshtein balls.
In what follows, we derive the minimum size of fixed-length $b$-burst Levenshtein balls via the size of $b$-burst-insertion balls, providing a new perspective on this problem.

\begin{theorem}\label{thm:minimum}
Let \(n \geq bt + 1\) with \(b, t \geq 1\). For any \(\bm{x} \in \Sigma_q^n\), it holds that
\begin{align*}
\big|\mathcal{L}_t^b(\bm{x})\big| \geq q^{t(b-1)} \sum\limits_{i=0}^t {n - t(b-1) \choose i} (q-1)^i,
\end{align*}
where the inequality holds with equality if and only if \(\big|\mathcal{D}_t^b(\bm{x})\big| = 1\). Recall from Lemma \ref{lem:del_ball} that when the number of \(b\)-runs in \(\bm{x}\) is exactly one, it holds that \(\big|\mathcal{D}_t^b(\bm{x})\big| = 1\).
\end{theorem}

\begin{IEEEproof}
Observe that for any \( \bm{y} \in \mathcal{D}_t^b(\bm{x}) \), it holds that \( \mathcal{I}_t^b(\bm{y}) \subseteq \mathcal{L}_t^b(\bm{x}) \). By Lemma \ref{lem:ins_ball}, we can compute
\[
\big|\mathcal{L}_t^b(\bm{x})\big| \geq \big|\mathcal{I}_t^b(\bm{y})\big| = q^{t(b-1)}\sum_{i=0}^t {n-t(b-1) \choose i} (q-1)^i.
\]
Since \( \mathcal{L}_t^b(\bm{x}) = \bigcup_{\bm{y} \in \mathcal{D}_t^b(\bm{x})} \mathcal{I}_t^b(\bm{y}) \), it remains to show that for \( \bm{y}' \neq \bm{y} \in \mathcal{D}_t^b(\bm{x}) \), there exists some \( \bm{z} \in \mathcal{I}_t^b(\bm{y}') \) such that \( \bm{z} \notin \mathcal{I}_t^b(\bm{y}) \).
Let \( j \) be the first position in which \( \bm{y} \) and \( \bm{y}' \) differ. Consider the sequence \( \bm{z} = \bm{y}_{[1,j-1]}' (y_j' 0^{b-1})^t \bm{y}_{[j,n]}' \). It can be easily verified that \( \bm{z} \in \mathcal{I}_t^b(\bm{y}') \) and \( \bm{z} \notin \mathcal{I}_t^b(\bm{y}) \). This completes the proof.
\end{IEEEproof}

\subsection{Explicit Size}

Sala and Dolecek~\cite{Sala-13-ISIT-ball} derived an exact closed-form expression for the size of a fixed-length Levenshtein ball with unit radius by analyzing the intersection of two insertion balls with a radius of one.

\begin{lemma}\cite[Theorem 1]{Sala-13-ISIT-ball}
Let \(n \geq 2\). For any \(\bm{x} \in \Sigma_q^n\), it holds that
\begin{equation}\label{eq:b=1}
\big|\mathcal{L}_1^1(\bm{x})\big|
= \big(n(q-1)-1\big) \cdot r(\bm{x}) + 2 - \sum\limits_{i=1}^{\alpha(\bm{x})} s_i(\bm{x}),
\end{equation}
where \(r(\bm{x})\) denotes the number of runs in \(\bm{x}\), \(\alpha(\bm{x})\) denotes the number of alternating segments in \(\bm{x} \), and \(s_i(\bm{x})\) denotes the length of the \(i\)-th alternating segment in \(\bm{x}\).
\end{lemma}

For any integer $b$, Lemma \ref{num_insertion} notes that the maximum intersection size between two $b$-burst-insertion balls with radius one grows exponentially with respect to $b$, complicating the determination of the explicit size of a fixed-length $b$-burst Levenshtein ball. In this paper, we propose a systematic approach to calculating this quantity.
Particularly, for the case where $b = 1$, our result is simpler than that of Sala and Dolecek~\cite{Sala-13-ISIT-ball}. In our expression, the total number of summands in each term does not depend on $\bm{x}$, whereas in their work, the term $\sum_{i=1}^{\alpha(\bm{x})} s_i(\bm{x})$ is more complex since both every summand and the overall number of summands depend on $\bm{x}$.

\begin{theorem}\label{thm:explicit}
    Let $n\geq b+1$ with $b\geq 1$, for any $\bm{x}\in\Sigma_q^n$, it holds that
    \begin{flalign}\label{E_1_1_equ}
    \big|\mathcal{L}_1^b(\bm{x})\big|
    &=q^{b-1}\big((n-b+1)(q-1)-1\big) \cdot r_b(\bm{x})+2q^{b-1}-\sum\limits_{j=1}^{b-1} q^{b
    -j-1}\cdot  f_{b,j}(\bm{x})-\sum\limits_{i=1}^{n-2b} g_{b,i}(\bm{x}),
    \end{flalign}
    where $r_b(\bm{x})$ denotes the number of $b$-runs in $\bm{x}$, $f_{b,j}(\bm{x})\triangleq \big|\{i\in[1,n-b-j]:x_i\neq x_{i+b},x_{i+j}\neq x_{i+j+b}\} \big|$ for $j\in[b-1]$, and $g_{b,i}(\bm{x})\triangleq \big| \{j\in [i+b,n-b]:x_i\ne x_{i+b},x_{j}\ne x_{j+b},\bm{x}_{[i,j+b]}\text{ is $2b$-periodic}\} \big|$.
\end{theorem}

In the rest of this paper, for simplicity of notation, we assume \(n \geq 2b + 1\) so that the quantities \(f_{b,j}(\bm{x})\) and \(g_{b,i}(\bm{x})\) are meaningful.

Building upon this explicit size, we can investigate the extremal size and concentration properties of the fixed-length $b$-burst Levenshtein balls, similar to what has been done for the fixed-length Levenshtein balls in the references \cite{Bar-Lev-22-IT-ball,Wang-24-DCC-ball,He-23-ISIT-ball}.

\subsection{Maximum Size}

For a non-binary alphabet, Bar-Lev et al.~\cite{Bar-Lev-22-IT-ball} established that when the center \(\bm{x} \in \Sigma_q^n\) satisfies \(x_i,x_{i+1},x_{i+2}\) are pairwise distinct for \(i \in [1, n-2]\), its fixed-length Levenshtein ball with unit radius achieves the maximum size among all fixed-length Levenshtein balls with unit radius. In this work, we generalize this contribution to fixed-length $b$-burst Levenshtein balls by requiring that the center \(\bm{x} \in \Sigma_q^n\) satisfies \(x_i \neq x_{i+b}\) and \(x_i \neq x_{i+2b}\) for \(i \in [1, n-2b]\), for which we formalize as follows.

\begin{theorem}\label{thm:maximum}
Let $n\geq 2b+1$ with $b\geq 2$, for any $\bm{x}\in \Sigma_q^n$, it holds that
  \begin{equation}\label{eq:max}
    \big|\mathcal{L}_1^b(\bm{x})\big|\leq q^{b-1}(n-b+1)\big((n-b+1)(q-1)+1\big)= q^{b-1}(q-1)n^2+O(n).
  \end{equation}
When $q\geq 3$, we have 
\begin{align*}
\big|\mathcal{L}_1^b(\bm{x})\big|
&\leq q^{b-1}(n-b+1)\big((n-b+1)(q-1)-1\big)+2q^{b-1}\\
&\quad\quad\quad -\frac{(n-b)(q^{b-1}-1)}{q-1}+ \frac{q^{b}-q-(b-1)(q-1)}{(q-1)^2}\\
&= q^{b-1}(q-1)n^2+O(n),
\end{align*}
where the inequality holds for equality when $x_i\ne x_{i+b}$ for $i\in[1,n-b]$ and $x_{i}\ne x_{i+2b}$ for $i\in[1,n-2b]$.
\end{theorem}

In the case of a binary alphabet, the problem becomes more complex because it is impossible to find a sequence $\bm{x} \in \Sigma_2^n$ where the elements $x_i, x_{i+1}, x_{i+2}$ are pairwise distinct for $i \in [1, n-2]$.
In \cite{Bar-Lev-22-IT-ball}, Bar-Lev et al. addressed the binary case by introducing the concept of an \(m\)-balanced sequence \(\bm{x}\), for which \(\alpha(\bm{x}) = m\) and the length of each alternating segment belongs to the set \(\big\{\lfloor n/m \rfloor, \lceil n/m \rceil\big\}\).
They demonstrated that
\[
\mathop{\arg\max}_{\substack{\bm{x} \in \Sigma_2^n \\ \alpha(\bm{x}) = m}} \big\{ \big|\mathcal{L}_1^1(\bm{x})\big| \big\}= \big\{\bm{x} \in \Sigma_2^n: \bm{x} \text{ is an \(m\)-balanced sequence}\big\}.
\]
Let
\[
M = \mathop{\arg\min}_{m \in [1,n]} \Big\{\Big|m - \frac{1}{2}\sqrt{1 + 2n}\Big|\Big\}.
\]
They further showed that the maximum fixed-length Levenshtein balls with unit radius are centered at the \(m\)-balanced sequences for \(m \in M\) and established that
\[
\mathop{\max}_{\substack{\bm{x} \in \Sigma_2^n}} \big\{ \big|\mathcal{L}_1^1(\bm{x})\big| \big\} = n^2 - \sqrt{2} n^{\frac{3}{2}} + O(n).
\]
Following a similar method as that used in \cite[Section IV.B]{Bar-Lev-22-IT-ball}, we can derive that
\[
\mathop{\max}_{\substack{\bm{x} \in \Sigma_2^n}} \big\{ \big|\mathcal{L}_1^b(\bm{x})\big| \big\} = 2^{b-1} n^2 - 2^{\frac{b}{2}} n^{\frac{3}{2}} + O(n).
\]
Since the discussion is more complicated and somewhat convoluted, and this result provides only a small improvement compared to our general upper bound \(2^{b-1} n^2 + O(n)\) established in Equation (\ref{eq:max}), we omit further details.

\subsection{Average Size}

For the special case where \(b = 1\), when using Equation (\ref{eq:b=1}), one needs to determine the expectations \(\mathop{\mathbb{E}}_{\bm{x} \in \Sigma_q^n}\left[ r(\bm{x}) \right]\) and \(\mathop{\mathbb{E}}_{\bm{x} \in \Sigma_q^n}\left[ \sum_{i=1}^{\alpha(\bm{x})} s_i(\bm{x}) \right]\). Computing \(\mathop{\mathbb{E}}_{\bm{x} \in \Sigma_q^n}\left[ \sum_{i=1}^{\alpha(\bm{x})} s_i(\bm{x}) \right]\) is a challenging task, see \cite[Section V]{Bar-Lev-22-IT-ball} for further details.
In contrast, when using Equation (\ref{E_1_1_equ}) for an arbitrary integer \(b\), it suffices to compute \(\mathop{\mathbb{E}}_{\bm{x} \in \Sigma_q^n}\left[ r_b(\bm{x}) \right]\), \(\mathop{\mathbb{E}}_{\bm{x} \in \Sigma_q^n}\left[ f_{b,j}(\bm{x}) \right]\), and \(\mathop{\mathbb{E}}_{\bm{x} \in \Sigma_q^n}\left[ g_{b,i}(\bm{x}) \right]\), which will be shown to be a relatively straightforward task in Section \ref{sec:average}.
After deriving \(\mathop{\mathbb{E}}_{\bm{x} \in \Sigma_q^n}\left[ r_b(\bm{x}) \right]\), \(\mathop{\mathbb{E}}_{\bm{x} \in \Sigma_q^n}\left[ f_{b,j}(\bm{x}) \right]\), and \(\mathop{\mathbb{E}}_{\bm{x} \in \Sigma_q^n}\left[ g_{b,i}(\bm{x}) \right]\), we can determine \(\mathop{\mathbb{E}}_{\bm{x} \in \Sigma_q^n}\left[\big|\mathcal{L}_1^b(\bm{x})\big|\right]\).

\begin{theorem}\label{thm:average}
Let $n\geq 2b+1$ with $b\geq 1$, it holds that
\begin{align*}
\mathop{\mathbb{E}}_{\bm{x}\in\Sigma_q^n}\big[\big|\mathcal{L}_1^b(\bm{x})\big|\big]
&=q^{b-2}\big(q+(q-1)(n-b)\big) \big((n-b+1)(q-1)-1\big) +2q^{b-1} \\
&\quad\quad\quad -\sum\limits_{j=1}^{b-1} \frac{q^b(q-1)^2(n-b-j)}{q^{j+3}}-\sum\limits_{i=1}^{n-2b} \sum_{j=i+b}^{n-b} \frac{k_{i,j}}{q^{j+b-i-1}}\\
&= q^{b-2}(q-1)^2n^2+O(n),
\end{align*}
where $k_{i,j}=(q-1)q^{2b-1}$ if $b\mid(j-i)$ and $k_{i,j}=(q-1)^2q^{2b-2}$ otherwise.
\end{theorem}

\subsection{Concentration}
Wang and Wang \cite{Wang-24-DCC-ball} established the concentration bound on fixed-length Levenshtein balls with radius one by Azuma's inequality.
Later, He and Ye \cite{He-23-ISIT-ball} provided a simple proof of the concentration bound presented by Wang and Wang.
Now we refine the approach of He and Ye \cite{He-23-ISIT-ball} to construct a code $\mathcal{C} \subseteq \Sigma_q^n$ containing $(1-\frac{3}{n})q^n$ sequences, such that the size of each fixed-length $b$-burst Levenshtein ball with radius one is close to the average size.
This yields the following concentration bound on fixed-length $b$-burst Levenshtein balls.

\begin{theorem}\label{thm:concentration}
  For $b\geq 1$ and sufficiently large $n$, there exists some constant $C$ such that for any $\bm{x}\in \Sigma_q^n$, it holds that
  \begin{align*}
    \mathbbm{P}\left( \Big| \big|\mathcal{L}_1^b(\bm{x})\big|- \mathop{\mathbb{E}}_{\bm{x} \in \Sigma_q^n}\big[\big|\mathcal{L}_1^b(\bm{x})\big|\big] \Big|> C\sqrt{n^3\log_q n}\right) \leq \frac{3}{n}.
  \end{align*}
\end{theorem}

\section{Proof of Theorem \ref{thm:explicit}}\label{sec:explicit}

Let \(\bm{x} \in \Sigma_q^n\) contain \(r_b(\bm{x})\) \(b\)-runs.
For \(i \in [1, r_b(\bm{x})]\), let \(p_i^s\) and \(p_i^e\) denote the starting and ending positions of the \(i\)-th \(b\)-run, respectively.
Additionally, define \(\bm{x}^i\) as the sequence obtained from \(\bm{x}\) by removing a substring of length \(b\) from its \(i\)-th \(b\)-run, and let \(X_i \triangleq \mathcal{I}_1^b(\bm{x}^i)\) denote the \(b\)-burst-insertion ball of \(\bm{x}^i\).
In the remainder of this section, the notations \(\bm{x}\), \(p_i^s\), \(p_i^e\), \(\bm{x}^i\), and \(X_i\) will correspond to those defined above.

By the definition of a \(b\)-run, we can conclude that
\begin{gather}
    p_i^e - p_i^s \geq b - 1 \quad \text{for } i \in [1, r_b(\bm{x})], \label{eq:position1} \\
    p_{i+1}^s = p_i^e - b + 2\quad \text{for } i \in [1, r_b(\bm{x})-1], \label{eq:position2}\\
    x_{p_i^s - 1} \neq x_{p_i^s - 1 + b} \quad \text{for } i \in [2, r_b(\bm{x})], \label{eq:del1}\\
    x_{p_i^e + 1} \neq x_{p_i^e + 1 - b}  \quad \text{for } i \in [1, r_b(\bm{x}) - 1]. \label{eq:del2}\\
    x_{t}= x_{t+b} \quad \text{for } t\not\in \{p_i^e+1-b: i\in [1,r_b(\bm{x})-1]\}. \label{eq:del3}
\end{gather}

The following conclusion can be directly derived from Lemmas \ref{lem:del_ball}, \ref{del_ball}, and \ref{lem:ins_ball}.

\begin{lemma}\label{lem:del}
    It holds that \(\mathcal{D}_1^b(\bm{x}) = \{\bm{x}^1, \bm{x}^2, \ldots, \bm{x}^{r_b(\bm{x})}\}\) and \(|\mathcal{L}_1^b(\bm{x})| = \big| \bigcup_{i=1}^{r_b(\bm{x})} X_i \big| \), where for any \(i \in [1, r_b(\bm{x})]\), we have \(\bm{x}^i = \bm{x}_{[n] \setminus [t, t + b - 1]}\) for each \(t \in [p_i^s, p_i^e - b + 1]\) and $|X_i|= q^{b-1}\big((n-b+1)(q-1)+1\big)$.
\end{lemma}

To derive a closed-form expression for \(|\mathcal{L}_1^b(\bm{x})|\), we will require the following three claims, whose proofs are deferred to Subsections \ref{subsec:cla1}, \ref{subsec:cla2}, and \ref{subsec:cla3}, respectively, concerning the sets \(X_1, X_2, \ldots, X_{r_b(\bm{x})}\).

Before proceeding, we will define two types of sets.

\begin{definition}\label{def:X}
    For $1\leq i<j\leq r_b(\bm{x})$, we define $A_{i,j}$ and $B_{i,j}$ as follows:
    \begin{itemize}

        \item If $1\leq p_{j-1}^{e}- p_i^e+1\leq b$, let
        \begin{align*}
            A_{i,j}\triangleq \left\{\bm{x}_{[1,p_i^e-b]} \bm{x}_{[p_i^e-b+1,p_{j-1}^e-b+1]} \bm{v} \bm{x}_{[p_i^e+1,p_{j-1}^e+1]} \bm{x}_{[p_{j-1}^e+2,n]}:\bm{v}\in\Sigma_q^{b-(p_{j-1}^{e}- p_i^e+1)} \right\}\supseteq \{\bm{x}\}
        \end{align*}
        and
        \begin{align*}
            B_{i,j}\triangleq \left\{\bm{x}_{[1,p_i^e-b]} \bm{x}_{[p_i^e+1,p_{j-1}^e+1]} \bm{v} \bm{x}_{[p_i^e-b+1,p_{j-1}^e-b+1]} \bm{x}_{[p_{j-1}^e+2,n]}:\bm{v}\in\Sigma_q^{b-(p_{j-1}^{e}- p_i^e+1)} \right\}.
        \end{align*}

        \item If $p_{j-1}^{e}- p_i^e+1\geq b+1$, let
        \begin{align*}
            A_{i,j}\triangleq \left\{\bm{x}_{[1,p_i^e-b]}  \bm{x}_{[p_i^e-b+1,p_{i}^e]} \bm{x}_{[p_i^e+1,p_{j-1}^e+1]}\bm{x}_{[p_{j-1}^e+2,n]}\right\}= \{\bm{x}\}
        \end{align*}
        and
        \begin{align*}
            B_{i,j}\triangleq \left\{ \bm{x}_{[1,p_i^e-b]} \bm{x}_{[p_i^e+1,p_{i}^e+b]} \bm{x}_{[p_i^e-b+1,p_{j-1}^e-b+1]} \bm{x}_{[p_{j-1}^e+2,n]} \right\}
        \end{align*}
        if $\bm{x}_{[p_i^e+b+1,p_{j-1}^e+1]}= \bm{x}_{[p_i^e-b+1,p_{j-1}^e-2b+1]}$, i.e., $\bm{x}_{[p_i^e-b+1,p_{j-1}^e+1]}$ is $2b$-periodic, otherwise $B_{i,j}=\emptyset$.
    \end{itemize}
\end{definition}

\begin{claim}\label{cla1}
   For $1\leq i<j\leq r_b(\bm{x})$, we have $X_i\cap X_j= A_{i,j} \sqcup B_{i,j}$. In particular, $|X_i\cap X_{i+1}|=2q^{b-1}$.


\end{claim}

\begin{claim}\label{cla2}
For $1\leq i<j<k\leq r_b(\bm{x})$, the following holds:
\begin{enumerate}
    \item $A_{i,j}\cap B_{i,k}=A_{i,k}\cap B_{i,j}=\emptyset$;

    \item $A_{i,k}\subseteq A_{i,j}$ and $A_{i,k}\subseteq A_{j,k}$;

    \item $B_{i,k}\cap B_{i,j}= \emptyset$ and \(B_{i,k}\cap B_{j,k} = \emptyset\);


    \item $(X_i \cap X_k)\setminus (X_i\cap X_j)= (X_i \cap X_k)\setminus (X_j\cap X_k)= B_{i,k}$.
\end{enumerate}
\end{claim}

\begin{claim}\label{cla3}
It holds that
\begin{align*}
|\mathcal{L}_1^b(\bm{x})|
&=\sum\limits_{i=1}^{r_b(\bm{x})}|X_i|-\sum\limits_{i=1}^{r_b(\bm{x})-1}|X_i\cap X_{i+1}|-\sum\limits_{i=1}^{r_b(\bm{x})-2}\sum\limits_{j=i+2}^{r_b(\bm{x})}|(X_i\cap X_j)\setminus(X_i\cap X_{i+1})| \\
&=\sum\limits_{i=1}^{r_b(\bm{x})}|X_i|-\sum\limits_{i=1}^{r_b(\bm{x})-1}|X_i\cap X_{i+1}|-\sum\limits_{i=1}^{r_b(\bm{x})}\sum\limits_{j=i+2}^{r_b(\bm{x})}|(X_i\cap X_j)\setminus(X_i\cap X_{i+1})|.
\end{align*}
\end{claim}

Now, using the above claims, we can calculate
\begin{align*}
\left|\mathcal{L}_1^b(\bm{x})\right|
&=\sum\limits_{i=1}^{r_b(\bm{x})}|X_i|-\sum\limits_{i=1}^{r_b(\bm{x})-1}|X_i\cap X_{i+1}|-\sum\limits_{i=1}^{r_b(\bm{x})}\sum\limits_{j=i+2}^{r_b(\bm{x})}|(X_i\cap X_j)\setminus(X_i\cap X_{i+1})|\\
&=r_b(\bm{x}) \cdot q^{b-1}\big((n-b+1)(q-1)+1\big)- \big( r_b(\bm{x})-1 \big)\cdot 2q^{b-1}- \sum\limits_{i=1}^{r_b(\bm{x})}\sum\limits_{j=i+2}^{r_b(\bm{x})} |B_{i,j}|.
\end{align*}
Since \( p_{i+1}^e - p_i^e \geq 1 \), it follows that \( p_{j-1}^{e} - p_i^e + 1 \geq 2 \) for \( j \geq i + 2 \). Recall by Definition \ref{def:X} that:
\begin{itemize}
    \item If \( 2 \leq p_{j-1}^{e} - p_i^e + 1 \leq b \), then \(|B_{i,j}| = q^{b - (p_{j-1}^{e} - p_i^e + 1)}\);
    \item If \( p_{j-1}^{e} - p_i^e + 1 \geq b + 1 \), then \( |B_{i,j}| = 1 \) if \( \bm{x}_{[p_i^e - b + 1, p_{j-1}^e + 1]} \) is \( 2b \)-periodic, otherwise \( |B_{i,j}| = 0 \).
\end{itemize}
To this end, we can divide the task of computing \(\sum\limits_{i=1}^{r_b(\bm{x})}\sum\limits_{j=i + 2}^{r_b(\bm{x})} |B_{i,j}|\) into two types of subtasks based on the value of \(p_{j-1}^{e} - p_i^e + 1 \).

For fixed integer $d\in [2,b]$, we can compute
\begin{align*}
  &\quad \sum\limits_{i=1}^{r_b(\bm{x})}\sum\limits_{j=i + 2}^{r_b(\bm{x})} |B_{i,j}| \cdot \mathbbm{1}_{p_{j-1}^{e} - p_i^e + 1=d}\\
  &= q^{b-d} \cdot \sum\limits_{i=1}^{r_b(\bm{x})}\sum\limits_{j=i + 2}^{r_b(\bm{x})} \mathbbm{1}_{p_{j-1}^{e} - p_i^e + 1=d}\\
  &= q^{b-d} \cdot \sum\limits_{i=1}^{r_b(\bm{x})} \sum\limits_{i'=1}^{n} \mathbbm{1}_{i'= p_i^e-b+1} \sum\limits_{j=i + 2}^{r_b(\bm{x})} \sum\limits_{j'=1}^{n} \mathbbm{1}_{j'= p_{j-1}^e-b+1} \cdot \mathbbm{1}_{p_{j-1}^{e} - p_i^e + 1=d}\\
  &= q^{b-d} \cdot \sum\limits_{i=1}^{r_b(\bm{x})} \sum\limits_{i'=1}^{n} \mathbbm{1}_{i'= p_i^e-b+1} \sum\limits_{j=i + 2}^{r_b(\bm{x})} \sum\limits_{j'=i'+1}^{n} \mathbbm{1}_{j'= p_{j-1}^e-b+1} \cdot \mathbbm{1}_{j'-i' + 1=d}\\
  &\stackrel{(\ast)}{=} q^{b-d} \cdot \sum\limits_{i'=1}^{n} \mathbbm{1}_{x_{i'}\neq x_{i'+b}} \sum\limits_{j'=i'+ 1}^{n} \mathbbm{1}_{x_{j'}\neq x_{j'+b}} \cdot \mathbbm{1}_{j' - i' + 1=d}\\
  &= q^{b-d} \cdot \sum\limits_{i'=1}^{n} \mathbbm{1}_{x_{i'}\neq x_{i'+b}} \cdot \mathbbm{1}_{x_{i'+d-1}\neq x_{i'+b+d-1}}\\
  &=q^{b-d} \cdot \sum\limits_{i'=1}^{n-b-d+1} \mathbbm{1}_{x_{i'}\neq x_{i'+b}} \cdot \mathbbm{1}_{x_{i'+d-1}\neq x_{i'+b+d-1}}\\
  &= q^{b-d} \cdot f_{b,d-1}(\bm{x}),
\end{align*}
where $(\ast)$ holds by Equations (\ref{eq:del2}) and (\ref{eq:del3}).
Similarly, we can compute
\begin{align*}
  &\quad\sum\limits_{i=1}^{r_b(\bm{x})}\sum\limits_{j=i + 2}^{r_b(\bm{x})} |B_{i,j}| \cdot \mathbbm{1}_{p_{j-1}^{e} - p_i^e + 1\geq b+1}\\
  &=\sum\limits_{i=1}^{r_b(\bm{x})}\sum\limits_{j=i + 2}^{r_b(\bm{x})} \mathbbm{1}_{p_{j-1}^{e} - p_i^e + 1\geq b+1} \cdot \mathbbm{1}_{\bm{x}_{[p_i^e - b + 1, p_{j-1}^e + 1]} \text{ is } 2b\text{-periodic}}\\
  &=\sum\limits_{i=1}^{r_b(\bm{x})} \sum\limits_{i'=1}^{n} \mathbbm{1}_{i'= p_i^e-b+1} \sum\limits_{j=i + 2}^{r_b(\bm{x})} \sum\limits_{j'=i'+1}^{n} \mathbbm{1}_{j'= p_{j-1}^e-b+1} \cdot \mathbbm{1}_{p_{j-1}^{e} - p_i^e + 1\geq b+1} \cdot \mathbbm{1}_{\bm{x}_{[p_i^e - b + 1, p_{j-1}^e + 1]} \text{ is } 2b\text{-periodic}}\\
  &=\sum\limits_{i=1}^{r_b(\bm{x})} \sum\limits_{i'=1}^{n} \mathbbm{1}_{i'= p_i^e-b+1} \sum\limits_{j=i + 2}^{r_b(\bm{x})} \sum\limits_{j'=i'+1}^{n} \mathbbm{1}_{j'= p_{j-1}^e-b+1} \cdot \mathbbm{1}_{j'-i' + 1\geq b+1} \cdot \mathbbm{1}_{\bm{x}_{[i',j'+b]} \text{ is } 2b\text{-periodic}}\\
  &\stackrel{(\ast)}{=} \sum\limits_{i'=1}^{n} \mathbbm{1}_{x_{i'}\neq x_{i'+b}} \sum\limits_{j'=i'+ 1}^{n} \mathbbm{1}_{x_{j'}\neq x_{j'+b}} \cdot \mathbbm{1}_{j'-i' + 1\geq b+1} \cdot \mathbbm{1}_{\bm{x}_{[i',j'+b]} \text{ is } 2b\text{-periodic}}\\
  &=\sum\limits_{i'=1}^{n-2b} \sum\limits_{j'=i'+ b}^{n-b} \mathbbm{1}_{x_{i'}\neq x_{i'+b}} \cdot \mathbbm{1}_{x_{j'}\neq x_{j'+b}} \cdot \mathbbm{1}_{\bm{x}_{[i',j'+b]} \text{ is } 2b\text{-periodic}}\\
  &= \sum\limits_{i'=1}^{n-2b} g_{b,i'}(\bm{x}),
\end{align*}
where $(\ast)$ holds by Equations (\ref{eq:del2}) and (\ref{eq:del3}).
Combining these with the fact that
\begin{align*}
  \sum\limits_{i=1}^{r_b(\bm{x})}\sum\limits_{j=i + 2}^{r_b(\bm{x})} |B_{i,j}|
  &= \sum_{d=2}^b \sum\limits_{i=1}^{r_b(\bm{x})}\sum\limits_{j=i + 2}^{r_b(\bm{x})} |B_{i,j}| \cdot \mathbbm{1}_{p_{j-1}^{e} - p_i^e + 1=d}+ \sum\limits_{i=1}^{r_b(\bm{x})}\sum\limits_{j=i + 2}^{r_b(\bm{x})} |B_{i,j}| \cdot \mathbbm{1}_{p_{j-1}^{e} - p_i^e + 1\geq b+1},
\end{align*}
the proof of Theorem \ref{thm:explicit} is completed.

\subsection{Proof of Claim \ref{cla1}}\label{subsec:cla1}

By Lemma \ref{lem:del} and Equation (\ref{eq:position2}), for $1\leq i<j\leq r_b(\bm{x})$, we have
    \begin{gather*}
        \bm{x}^i= \bm{x}_{[n]\setminus [p_i^e-b+1,p_i^e]}= \bm{x}_{[1,p_i^e-b]} \bm{x}_{[p_i^e+1,n]}= \bm{x}_{[1,p_i^e-b]} \bm{x}_{[p_i^e+1,p_{j-1}^e+1]} \bm{x}_{[p_{j-1}^e+2,n]},\\
        \bm{x}^j= \bm{x}_{[n]\setminus [p_{j}^s,p_{j}^s+b-1]}= \bm{x}_{[1,p_{j}^s-1]} \bm{x}_{[p_{j}^s+b,n]}= \bm{x}_{[1,p_{j-1}^e-b+1]} \bm{x}_{[p_{j-1}^e+2,n]}= \bm{x}_{[1,p_i^e-b]} \bm{x}_{[p_i^e-b+1,p_{j-1}^e-b+1]} \bm{x}_{[p_{j-1}^e+2,n]} .
    \end{gather*}
    By Equation (\ref{eq:del2}), we get $x_{p_i^e+1}\neq x_{p_i^e-b+1}$ and $x_{p_{j-1}^e+1}\neq x_{p_{j-1}^e-b+1}$.
    Then the conclusion follows by Lemma \ref{num_insertion}.

\subsection{Proof of Claim \ref{cla2}}\label{subsec:cla2}

For the first statement, since by Equation (\ref{eq:del2}) that $x_{p_i^e + 1} \neq x_{p_i^e + 1 - b}$, it can be inferred from Claim \ref{cla1} that $A_{i,j}\cap B_{i,k}=A_{i,k}\cap B_{i,j}=\emptyset$.

For the second statement, if $A_{i,k}=\{\boldsymbol{x}\}$, it is clear that $A_{i,k}\subseteq A_{i,j}$.
If $A_{i,k}\neq \{\boldsymbol{x}\}$, we have $2\leq p_{j-1}^{e}- p_i^e+1< p_{k-1}^{e}- p_i^e+1\leq b$.
Then we get
\begin{gather*}
    A_{i,j}= \left\{\bm{x}_{[1,p_{j-1}^e-b+1]} \bm{v} \bm{x}_{[p_i^e+1,n]}:\bm{v}\in\Sigma_q^{b-(p_{j-1}^{e}- p_i^e+1)} \right\},\\
    A_{i,k}= \left\{\bm{x}_{[1,p_{k-1}^e-b+1]} \bm{v} \bm{x}_{[p_i^e+1,n]}:\bm{v}\in\Sigma_q^{b-(p_{k-1}^{e}- p_i^e+1)} \right\}.
\end{gather*}
In this case, we also have $A_{i,k}\subseteq A_{i,j}$.
Similarly, by considering the reversal of $\bm{x}$, we can derive $A_{i,k}\subseteq A_{j,k}$.

For the third statement, if \(B_{i,j} = \emptyset\) or \(B_{i,k} = \emptyset\), it is clear that \(B_{i,k} \cap B_{i,j} = \emptyset\).
If \(B_{i,j} \neq \emptyset\) and \(B_{i,k} \neq \emptyset\), we note that the \((n - p_{k-1}^e)\)-th entry from the end of each sequence in \(B_{i,j}\) is \(x_{p_{k-1}^e + 1}\), while for \(B_{i,k}\), it is \(x_{p_{k-1}^e - b + 1}\). By Equation (\ref{eq:del2}), it follows that \(B_{i,k}\cap B_{i,j} = \emptyset\).
Similarly, by considering the reversal of $\bm{x}$, we can derive \(B_{i,k}\cap B_{j,k} = \emptyset\).

For the last statement, we can calculate
\begin{align*}
    (X_{i}\cap X_k)\setminus (X_i\cap X_j)
    &= (A_{i,k}\sqcup B_{i,k})\setminus (A_{i,j}\sqcup B_{i,j})\\
    &= \big(A_{i,k}\setminus (A_{i,j}\sqcup B_{i,j}) \big) \sqcup  \big(B_{i,k}\setminus (A_{i,j}\sqcup B_{i,j}) \big)\\
    &= B_{i,k},
\end{align*}
where the last equality follows by the first three statements.
Similarly, by considering the reversal of $\bm{x}$, we can derive $(X_{i}\cap X_k)\setminus (X_j\cap X_k)= B_{i,k}$.
This completes the proof.

\subsection{Proof of Claim \ref{cla3}}\label{subsec:cla3}

We prove it by induction on $r_b(\bm{x})$.
For the base case where $r_b(\bm{x})=1$, the conclusion is trivial. Now, assuming that the conclusion is valid for $r_b(\bm{x})\leq m$ with $m\geq 1$, we examine the scenario where $r_b(\bm{x})=m+1$ and calculate
\begin{align*}
\left|\bigcup\limits_{i=1}^{m+1}X_i\right|
&\stackrel{(\ast)}{=}\left|\bigcup\limits_{i=1}^{m}X_i \right|+|X_{m+1}|-\left|\bigcup\limits_{i=1}^{m} (X_i\cap X_{m+1}) \right|\\
&= \left|\bigcup\limits_{i=1}^{m}X_i \right|+|X_{m+1}|-|X_m\cap X_{m+1}|- \left|\bigcup\limits_{i=1}^{m-1} \big((X_i\cap X_{m+1}) \setminus (X_m\cap X_{m+1})\big)\right|\\
&\stackrel{(\star)}{=}\sum\limits_{i=1}^{m+1}|X_i|-\sum\limits_{i=1}^{m}|X_i\cap X_{i+1}|-\sum\limits_{i=1}^{m-2}\sum\limits_{j=i+2}^m|(X_i\cap X_j)\setminus(X_i\cap X_{i+1})|- \left|\bigcup\limits_{i=1}^{m-1} \big((X_i\cap X_{m+1}) \setminus (X_m\cap X_{m+1})\big)\right|\\
&\stackrel{(\diamond)}{=}\sum\limits_{i=1}^{m+1}|X_i|-\sum\limits_{i=1}^{m}|X_i\cap X_{i+1}|-\sum\limits_{i=1}^{m-1}\sum\limits_{j=i+2}^{m+1}|(X_i\cap X_j)\setminus(X_i\cap X_{i+1})|,
\end{align*}
where $(\ast)$ follows by the inclusion-exclusion principle, $(\star)$ follows by the induction hypothesis, and $(\diamond)$ follows by the last two statements of Claim \ref{cla2}.
Consequently, the conclusion is also valid for $r_b(\bm{x})=m+1$.
This completes the proof.

\section{Proof of Theorem \ref{thm:maximum}}\label{sec:maximum}

Since \( \mathcal{L}_t^b(\bm{x}) = \bigcup_{\bm{y} \in \mathcal{D}_t^b(\bm{x})} \mathcal{I}_t^b(\bm{y}) \), by Lemmas \ref{del_ball} and \ref{lem:ins_ball}, we can compute
\begin{align*}
  \big| \mathcal{L}_t^b(\bm{x}) \big|
  &\leq \sum_{\bm{y} \in \mathcal{D}_t^b(\bm{x})} \big| \mathcal{I}_t^b(\bm{y}) \big|
  \leq q^{b-1}(n-b+1)\big((n-b+1)(q-1)+1\big).
\end{align*}
When $q\geq 3$, we define
\begin{align*}
  h(\bm{x})\triangleq q^{b-1}\big((n-b+1)(q-1)-1\big) \cdot r_b(\bm{x})+2q^{b-1}-\sum\limits_{j=1}^{b-1} q^{b-j-1}\cdot  f_{b,j}(\bm{x}).
\end{align*}
It is clear that $\left|\mathcal{L}_1^b(\bm{x})\right|= h(\bm{x})- \sum_{i=1}^{n-2b} g_{b,i}(\bm{x})\leq h(\bm{x})$.
We then prove Theorem \ref{thm:maximum} through the following steps:
\begin{itemize}
  \item Demonstrating that \( h(\bm{x}) \) is a strictly increasing function of \( r_b(\bm{x})\) and that \( h(\bm{x}) \) is a constant independent on the choice of $\bm{x}$ when $r_b(\bm{x})$ achieves its maximum value $n-b+1$;
  \item Finding a sequence $\bm{x}$ such that $r_b(\bm{x})=n-b+1$ and $\sum_{i=1}^{n-2b} g_{b,i}(\bm{x})=0$.
\end{itemize}

\begin{claim}\label{cla:h}
  Let $n\geq 2b+1$ with $b\geq 2$, for any $\bm{x}, \bm{x}' \in \Sigma_q^n$ with $r_b(\bm{x})> r_b(\bm{x}')$, it holds that $h(\bm{x})> h(\bm{x}')$.
  Moreover, let
  \begin{align*}
  h\triangleq q^{b-1}(n-b+1)\big((n-b+1)(q-1)-1\big)+2q^{b-1}-\frac{(n-b)(q^{b-1}-1)}{q-1}+ \frac{q^{b}-q-(b-1)(q-1)}{(q-1)^2},
  \end{align*}
  we have $\max_{\bm{x}\in \Sigma_q^n}\{h(\bm{x})\}= h$ and this maximum size can be achieved if and only if $r_b(\bm{x})=n-b+1$.
\end{claim}

\begin{IEEEproof}
Recall that $f_{b,j}(\bm{x})\triangleq \big|\{i\in[1,n-b-j]:x_i\neq x_{i+b},x_{i+j}\neq x_{i+j+b}\} \big|\leq n-b-j$ for $j\in[1,b-1]$, we can compute
\begin{equation}\label{eq:f}
\begin{aligned}
  \sum\limits_{j=1}^{b-1} q^{b-j-1} \cdot  f_{b,j}(\bm{x})
  &\stackrel{(\ast)}{\leq} \sum\limits_{j=1}^{k} q^{b-j-1} \cdot (n-b-j)\\
  &= (n-b) \cdot \sum\limits_{j=1}^{b-1} q^{b-j-1}- \sum\limits_{j=1}^{b-1} j\cdot q^{b-j-1}\\
  &=\frac{(n-b)(q^{b-1}-1)}{q-1}- \frac{q^{b}-q-(b-1)(q-1)}{(q-1)^2}\\
  &< q^{b-1}\big((n-b+1)(q-1)-1\big),
\end{aligned}
\end{equation}
where $(\ast)$ holds for equality if and only if $x_i\neq x_{i+b}$ for $i\in [1,n-b]$, or equivalently $r_b(\bm{x})=n-b+1$.

Since $r_b(\bm{x})> r_b(\bm{x}')$, it then follows by definition that
\begin{align*}
  h(\bm{x})
  &> q^{b-1}\big((n-b+1)(q-1)-1\big) \cdot \big(r_b(\bm{x}')+1\big)+2q^{b-1}- q^{b-1}\big((n-b+1)(q-1)-1\big)\\
  &= q^{b-1}\big((n-b+1)(q-1)-1\big) \cdot r_b(\bm{x}')+2q^{b-1}\\
  &\geq h(\bm{x}').
\end{align*}
When $r_b(\bm{x})=n-b+1$, by Equation (\ref{eq:f}), we can compute $h(\bm{x})=h$. Note that $r_b(\bm{x})\leq n-b+1$, the proof is completed.
\end{IEEEproof}

When $q\geq 3$, there exists some sequence $\bm{x}\in \Sigma_q^n$ such that $x_i, x_{i+b}, x_{i+2b}$ are pairwise distinct for $i\in [1,n-2b]$. In this case, it follows by definition that $r_b(\bm{x})=n-b+1$ and $g_{b,i}(\bm{x})=0$ for $i\in [1,n-2b]$.
Since $\left|\mathcal{L}_1^b(\bm{x})\right|= h(\bm{x})- \sum_{i=1}^{n-2b} g_{b,i}(\bm{x})$, by Claim \ref{cla:h}, the conclusion of Theorem \ref{thm:maximum} follows.

\section{Proof of Theorem \ref{thm:average}}\label{sec:average}
To determine the expected value $\mathop{\mathbb{E}}_{\bm{x}\in\Sigma_q^n}\big[\big|\mathcal{L}_1^b(\bm{x})\big|\big]$, we need to compute the expected values $\mathop{\mathbb{E}}_{\bm{x}\in\Sigma_q^n}\big[r_b(\bm{x})\big]$, $\mathop{\mathbb{E}}_{\bm{x}\in\Sigma_q^n}\big[f_{b,j}(\bm{x})\big]$, and $\mathop{\mathbb{E}}_{\bm{x}\in\Sigma_q^n}\big[g_{b,i}(\bm{x})\big]$.

\begin{lemma}\label{lem:r}
Let $n\geq b+1$ with $b\geq 1$, it holds that
\begin{align*}
\mathop{\mathbb{E}}_{\bm{x}\in\Sigma_q^n}[r_b(\bm{x})]=1+\frac{(q-1)(n-b)}{q}.
\end{align*}
\end{lemma}

\begin{IEEEproof}
Recall by Lemma \ref{del_ball} that $r_b(\bm{x})= 1+ |\{i\in [1,n-b]: x_i\neq x_{i+b}\}| \leq n-b+1$.
For $i\in [1,n-b]$, we define $S_i\triangleq \{\bm{x}\in \Sigma_q^n: x_i\neq x_{i+b}\}$.
It is clear that $|S_i|=q^{n-1}(q-1)$.
Now, we can compute
\begin{align*}
    \mathop{\mathbb{E}}_{\bm{x}\in\Sigma_q^n}[r_b(\bm{x})]
    &=\frac{1}{q^n} \sum_{\bm{x}\in \Sigma_q^n} r_b(\bm{x})\\
    &=\frac{1}{q^n} \left(q^n+ \sum_{\bm{x}\in \Sigma_q^n}\sum_{i=1}^{n-b} \mathbbm{1}_{x_i\neq x_{i+b}} \right)\\
    &=1+\frac{1}{q^n} \sum_{i=1}^{n-b} \sum_{\bm{x}\in \Sigma_q^n}\mathbbm{1}_{x_i\neq x_{i+b}} \\
    &=1+\frac{1}{q^n} \sum_{i=1}^{n-b} |S_i|\\
    &= 1+\frac{(q-1)(n-b)}{q}.
\end{align*}
This completes the proof.
\end{IEEEproof}

\begin{lemma}\label{lem:f}
Let $n\geq 2b+1$ with $b\geq 1$, for $j\in [1,b-1]$, it holds that
\begin{align*}
\mathop{\mathbb{E}}_{\bm{x}\in\Sigma_q^n}[f_{b,j}(\bm{x})]=\frac{(q-1)^2(n-b-j)}{q^2}.
\end{align*}
\end{lemma}

\begin{IEEEproof}
Recall that $f_{b,j}(\bm{x})\triangleq \big|\{i\in[1,n-b-j]:x_i\neq x_{i+b},x_{i+j}\neq x_{i+j+b}\} \big|$.
For $i\in [1,n-b-j]$, we define $S_{i,j}= \{\bm{x}\in \Sigma_q^n: x_i\neq x_{i+b},x_{i+j}\neq x_{i+j+b}\}$.
It is clear that $|S_{i,j}|= q^{n-2}(q-1)^2$.
Now, we can compute
\begin{align*}
\mathop{\mathbb{E}}_{\bm{x}\in\Sigma_q^n}[f_{b,j}(\bm{x})]
&= \frac{1}{q^n}\sum_{\bm{x}\in \Sigma_q^n} f_{b,j}(\bm{x})\\
&= \frac{1}{q^n} \sum_{\bm{x}\in \Sigma_q^n}\sum_{i=1}^{n-b-j} \mathbbm{1}_{x_i\neq x_{i+b}, x_{i+j}\neq x_{i+j+b}}\\
&= \frac{1}{q^n} \sum_{i=1}^{n-b-j} \sum_{\bm{x}\in \Sigma_q^n} \mathbbm{1}_{x_i\neq x_{i+b}, x_{i+j}\neq x_{i+j+b}}\\
&= \frac{1}{q^n}\sum_{i=1}^{n-b-j} |S_{i,j}|\\
&= \frac{(q-1)^2(n-b-j)}{q^2}.
\end{align*}
This completes the proof.
\end{IEEEproof}

\begin{lemma}\label{lem:g}
Let $n\geq 2b+1$ with $b\geq 1$, for $i\in [1,n-2b]$, it holds that
\begin{flalign*}
\mathop{\mathbb{E}}_{\bm{x}\in\Sigma_q^n}[g_{b,i}(\bm{x})]
&=\sum\limits_{i=1}^{n-2b}\sum\limits_{j=i+b}^{n-b}k_{i,j}q^{-(j+b-i+1)},
\end{flalign*}
where $k_{i,j}=(q-1)q^{2b-1}$ when $b\mid(j-i)$, and $k_{i,j}=(q-1)^2q^{2b-2}$ when $b\nmid(j-i)$.
\end{lemma}

\begin{IEEEproof}
Recall that $g_{b,i}(\bm{x})\triangleq \big| \{j\in [i+b,n-b]:x_i\ne x_{i+b},x_{j}\ne x_{j+b},\bm{x}_{[i,j+b]}\text{ is $2b$-periodic}\} \big|$.
For $j\in [i+b,n-b]$, we define $S_{i,j}'= \{\bm{x}\in \Sigma_q^n: x_i\ne x_{i+b},x_{j}\ne x_{j+b},\bm{x}_{[i,j+b]}\text{ is $2b$-periodic}\}$.
\begin{itemize}
\item When $b\mid(j-i)$, since $\bm{x}_{[i,j+b]}$ is $2b$-periodic, we can conclude that $x_i\neq x_{i+b}$ if and only if $x_{j}\neq x_{j+b}$.
    In this case, we can write $S_{i,j}'= \{\bm{x}\in \Sigma_q^n: x_i\ne x_{i+b},\bm{x}_{[i,j+b]}\text{ is $2b$-periodic}\}$ and can compute $|S_{i,j}'|= q^{2b-1}(q-1) \cdot q^{n-(j+b-i+1)}$.
\item When $b\nmid(j-i)$, we can conclude that the conditions $x_i\neq x_{i+b}$ and $x_{j}\neq x_{j+b}$ are independent. In this case, we can compute $|S_{i,j}'|= q^{2b-2}(q-1)^2 \cdot q^{n-(j+b-i+1)}$.
\end{itemize}
Now, we get
\begin{align*}
\mathop{\mathbb{E}}_{\bm{x}\in\Sigma_q^n}[g_{b,i}(\bm{x})]
&=\frac{1}{q^n}\sum\limits_{\bm{x}\in\Sigma_q^n}g_{b,i}(\bm{x})\\
&= \frac{1}{q^n} \sum_{\bm{x}\in \Sigma_q^n}\sum_{j=i+b}^{n-b} \mathbbm{1}_{x_i\neq x_{i+b}, x_{j}\neq x_{j+b}, \bm{x}_{[i,j+b]} \text{ is $2b$-periodic}}\\
&= \frac{1}{q^n} \sum_{j=i+b}^{n-b} \sum_{\bm{x}\in \Sigma_q^n} \mathbbm{1}_{x_i\neq x_{i+b}, x_{j}\neq x_{j+b}, \bm{x}_{[i,j+b]} \text{ is $2b$-periodic}}\\
&=\frac{1}{q^n}\sum\limits_{j=i+b}^{n-b}|S_{i,j}'|\\
&=\sum\limits_{j=i+b}^{n-b} \frac{k_{i,j}}{q^{j+b-i+1}}.
\end{align*}
This completes the proof.
\end{IEEEproof}

Since
\begin{align*}
\mathop{\mathbb{E}}_{\bm{x}\in\Sigma_q^n} \big[\big|\mathcal{L}_1^b(\bm{x})\big|\big]
&=q^{b-1}\big((n-b+1)(q-1)-1\big) \cdot \mathop{\mathbb{E}}_{\bm{x}\in\Sigma_q^n}\big[r_b(\bm{x})\big]+2q^{b-1} \\
&\quad\quad -\sum\limits_{j=1}^{b-1} q^{b
-j-1}\cdot  \mathop{\mathbb{E}}_{\bm{x}\in\Sigma_q^n}\big[f_{b,j}(\bm{x})\big] -\sum\limits_{i=1}^{n-2b} \mathop{\mathbb{E}}_{\bm{x}\in\Sigma_q^n} \big[g_{b,i}(\bm{x})\big],
\end{align*}
the conclusion of Theorem \ref{thm:average} follows by Lemmas \ref{lem:r}, \ref{lem:f}, and \ref{lem:g}.

\section{Proof of Theorem \ref{thm:concentration}}\label{sec:concentration}

Our strategy is to construct a code $\mathcal{C} \subseteq \Sigma_q^n$ containing $(1-\frac{3}{n})q^n$ sequences, such that the size of each Levenshtein ball is close to the average size $\mathop{\mathbb{E}}_{\bm{x} \in \Sigma_q^n} \left[ \left| \mathcal{L}_1^b(\bm{x}) \right| \right]$.

\begin{construction}\label{constr}
Let $\mathcal{C}$ be a collection of sequences in $\Sigma_q^n$ such that each $\bm{x} \in \mathcal{C}$ satisfies the following conditions:
  \begin{itemize}
    \item The number of $b$-runs in $\bm{x}$ satisfies
    \begin{align*}
      \left(\frac{q-1}{q}-\sqrt{\frac{\ln n}{2(n-b)}} \right)(n-b) + 2 \leq r_b(\bm{x}) \leq \left(\frac{q-1}{q}+\sqrt{\frac{\ln n}{2(n-b)}} \right)(n-b).
    \end{align*}
    \item The length of each $2b$-periodic substring in $\bm{x}$ is less than $2\log_q n+2b$.
  \end{itemize}
\end{construction}

\begin{lemma}\label{lem:size}
  Let $\mathcal{C}$ be defined in Construction \ref{constr}, it holds that $|\mathcal{C}| \geq (1-\frac{3}{n})q^n$.
\end{lemma}

\begin{proof}
Choose a sequence $\boldsymbol{x}\in \Sigma_q^{n}$ uniformly at random, let $p_1$ be the probability that $r_b(\bm{x})$ does not belong to the interval $\left[\left(\frac{q-1}{q}-\sqrt{\frac{\ln n}{2(n-b)}}\right)+2, \left(\frac{q-1}{q}+\sqrt{\frac{\ln n}{2(n-b)}}\right)\right]$ and let $p_2$ be the probability that $\bm{x}$ contains a $2b$-periodic substring of length $2\log_q n+2b$, then we can compute $|\mathcal{C}|\geq (1-p_1-p_2)q^n$.
It remains to show that $p_1+p_2\leq \frac{3}{n}$.

We first provide an upper bound on the value of $p_1$.
Let $X_{n,p}$ be a random variable that obeys binomial distribution $B(n,p)$, i.e., $\mathbbm{P}(X_{n,b}=m)= \binom{n}{m}p^m(1-p)^{n-m}$ for $m\in [0,n]$.
Since by \cite[Lemma IV.4]{Lan-25} that $\big|\{\boldsymbol{x}\in \Sigma_q^n: r_b(\bm{x})=i+1 \} \big|= q^b(q-1)^{i}\binom{n-b}{i}= q^n \binom{n-b}{i}\big(\frac{q-1}{q}\big)^i \big(1-\frac{q-1}{q}\big)^{n-b-i}$ for $i\in [0,n-b]$, we can conclude that $r_b(\bm{x})-1$ (when regarding as a random variable) obeys binomial distribution $B\big(n-b,\frac{q-1}{q}\big)$.
Now, by using Hoeffding's inequality, we can compute
\begin{align*}
  p_1
  &= \mathbbm{P}\left(\left| r_b(\bm{x})-1- \frac{(q-1)(n-b)}{q} \right|\geq \sqrt{\frac{\ln n}{2(n-b)}}(n-b) \right) \\
  &= \mathbbm{P}\left(\left|X_{n-b,\frac{q-1}{q}}- \frac{(q-1)(n-b)}{q} \right|\geq \sqrt{\frac{\ln n}{2(n-b)}}(n-b) \right) \\
  &\leq \frac{2}{n}.
\end{align*}

Regarding $p_2$, for any $i\in [1,n-2\log_q n-2b+1]$, the probability that $\bm{x}_{[i,i+2\log_q n+2b-1]}$ is $2b$-periodic can be calculated as
\begin{equation*}
    p_2':= \frac{q^{2b}}{q^{2\log_q n+2b}}= \frac{1}{n^2}.
\end{equation*}
Then by the union bound, we can compute $p_2\leq n\cdot p_2' \leq \frac{1}{n}$.
Now, we derive $p_1+p_2=\frac{3}{n}$, which completes the proof.
\end{proof}

\begin{lemma}\label{lem:exp}
  Let $\mathcal{C}$ be defined in Construction \ref{constr}, there exists some constant $C$ such that for any $\bm{x}\in \mathcal{C}$, it holds that
  \begin{align*}
    \Big| \big|\mathcal{L}_1^b(\bm{x})\big|- \mathop{\mathbb{E}}_{\bm{x} \in \Sigma_q^n}\big[\big|\mathcal{L}_1^b(\bm{x})\big|\big] \Big| \leq C\sqrt{n^3\log_q n}.
  \end{align*}
\end{lemma}

\begin{proof}
Since the length of each $2b$-periodic substring in $\bm{x} \in \mathcal{C}$ is less than $2\log_q n+2b$, it follows by the definition that $g_{b,i}(\bm{x})\leq 2\log_q n$ for $i\in [1,n-2b]$.
This implies that $\sum\limits_{i=1}^{n-2b} g_{b,i}(\bm{x})\leq n\cdot 2\log_q n$.
Moreover, by Equation (\ref{eq:f}), we get $\sum\limits_{j=1}^{b-1} f_{b,j}(\bm{x}) q^{b-j-1}< q^{b-1} \big((n-b+1)(q-1)-1\big)$.
Since each $\bm{x} \in \mathcal{C}$ satisfies $r_b(\bm{x})= \frac{(q-1)n}{q}+O(\sqrt{n\log_q n})$, we can compute
\begin{align*}
q^{b-2}(q-1)^2n^2-c_1 \sqrt{n^3\log_q n} \leq \big|\mathcal{L}_1^b(\bm{x})\big|\leq q^{b-2}(q-1)^2n^2+c_2\sqrt{n^3\log_q n},
\end{align*}
for some constant $c_1$ and $c_2$.
Recall by Theorem \ref{thm:average} that $\mathop{\mathbb{E}}_{\bm{x} \in \Sigma_q^n}\big[\big|\mathcal{L}_1^b(\bm{x})\big|\big]= q^{b-2}(q-1)^2n^2+O(n)$, then the conclusion follows.
\end{proof}

The conclusion of Theorem \ref{thm:concentration} follows by Lemmas \ref{lem:size} and \ref{lem:exp} directly.


\end{document}